\documentclass[letterpaper, 10 pt, conference]{ieeeconf}  

\IEEEoverridecommandlockouts                              
\usepackage{times} 
\usepackage{enumerate}
\usepackage{mathtools}
\usepackage{amsmath} 
\usepackage{amssymb}  

\usepackage{amsthm}
\usepackage{graphicx}
\usepackage{xcolor}
\usepackage{url}
\usepackage{algorithm2e}
\usepackage{cite}

\usepackage{etoolbox}
\newtoggle{greyscale}
\settoggle{greyscale}{false}
\iftoggle{greyscale}
{%
	\graphicspath{{./figures/grey/}{./figures/}{./}}
}{
	\graphicspath{{./figures/paper/}{./figures/}{./}}
}

\usepackage{abbreviations}

\usepackage{siunitx,units}
\usepackage{booktabs, makecell, multirow}
\usepackage{aligned-overset}

\allowdisplaybreaks

\newtheorem{theorem}{Theorem}

\newtheorem{proposition}[theorem]{Proposition}
\newtheorem{definition}[theorem]{Definition}
\newtheorem{remark}[theorem]{Remark}
\newtheorem{assumption}[theorem]{Assumption}
\newtheorem{example}[theorem]{Example}

\newcommand{\bmatAB}{\begin{bmatrix}
	A & B
\end{bmatrix}}
\newcommand{\bmatABP}{\begin{bmatrix}
	A & B & P
\end{bmatrix}}
\newcommand{\bmatP}{\begin{bmatrix}
	\cP & 1
\end{bmatrix}}
\newcommand{\SigABP}{\Sigma_{A,B,P}}

\title{\LARGE \bf
Data-Driven Control of Nonlinear Systems: Beyond Polynomial Dynamics
}

\author{Robin Str\"asser, Julian Berberich, Frank Allg\"ower
\thanks{This work was funded by Deutsche Forschungsgemeinschaft (DFG, German Research Foundation) under Germany's Excellence Strategy - EXC 2075 - 390740016.
We acknowledge the support by the Stuttgart Center for Simulation Science (SimTech).
The authors thank the International Max Planck Research School for Intelligent Systems (IMPRS-IS) for supporting Julian Berberich.}
\thanks{Robin Str\"asser, Julian Berberich, and Frank Allg\"ower are with the Institute for Systems Theory and Automatic Control, University of Stuttgart, 70550 Stuttgart, Germany
		(email:\{robin.straesser, julian.berberich, frank.allgower\}@ist.uni-stuttgart.de).}%
}

\begin{document} 

\pubid{\begin{minipage}{\textwidth}\ \\[60pt] \copyright 2021 IEEE. This version has been accepted for publication in Proc. IEEE Conference on Decision and Control (CDC), 2021. Personal use of this material is permitted. Permission from IEEE must be obtained for all other uses, in any current or future media, including reprinting/republishing this material for advertising or promotional purposes, creating new collective works, for resale or redistribution to servers or lists, or reuse of any copyrighted component of this work in other works.\end{minipage}}
 
\maketitle


\begin{abstract}%
    In this paper, we present a data-driven controller design method for continuous-time nonlinear systems, using no model knowledge but only measured data affected by noise.
    While most existing approaches focus on systems with polynomial dynamics, our approach allows to design controllers for unknown systems with rational or general non-polynomial dynamics.
    We first derive a data-driven parametrization of unknown nonlinear systems with rational dynamics.
    By applying robust control techniques to this parametrization, we obtain sum-of-squares based criteria for designing controllers with closed-loop robust stability and performance guarantees for all systems which are consistent with the measured data and the assumed noise bound.
    We then apply this approach to control systems whose dynamics are linear in general non-polynomial basis functions by transforming them into polynomial systems.
    Finally, we apply the developed approaches to numerical examples.
\end{abstract}


\section{Introduction}
Many systems in natural and engineering tasks possess nonlinear dynamic components which cannot be neglected for controller design with good practical performance as well as rigorous guarantees. Therefore, a standard procedure is to first obtain an accurate nonlinear system model using either a first-principles approach or system identification and then apply established nonlinear controller design techniques~\cite{khalil:2002}. However, the identification of nonlinear systems can be challenging and time-consuming. Thus, learning and direct data-driven methods for nonlinear systems have received increasing attention in the last decades~\cite{hou:wang:2013}. 

For linear time-invariant (LTI) systems, the work in~\cite{willems:rapisarda:markovsky:demoor:2005} provides a promising framework for data-driven control. 
Based on this result, \cite{depersis:tesi:2019} design state- and output-feedback controllers using only measured data, \cite{waarde:camlibel:mesbahi:2020} provide an improvement for controller design based on noisy data, and~\cite{berberich:scherer:allgower:2020} develop a framework for combining data and possibly available prior knowledge for controller design.
Since all of these works focus on LTI systems, an immediate question is whether Willems et al.'s fundamental lemma~\cite{willems:rapisarda:markovsky:demoor:2005} and related works can also be extended to nonlinear systems. 
Extensions of the fundamental lemma have been developed for Hammerstein and Wiener systems~\cite{berberich:allgower:2020}, for second-order Volterra systems~\cite{rueda-escobedo:schiffer:2020}, and for flat nonlinear systems~\cite{alsalti:berberich:lopez:allgower:muller:2021}. 
The work in~\cite{lian:wang:jones:2021} uses the Koopman operator~\cite{koopman:neumann:1932} to first lift a nonlinear system to an infinite-dimensional linear system and then use linear design methods based on a finite-dimensional approximation, however, without any closed-loop guarantees.

A popular approach for model-based nonlinear control relies on sum-of-squares (SOS) optimization, which reformulates nonnegativity conditions as semi-definite programs (SDP)~\cite{parrilo:2000}.
In~\cite{martin:allgower:2020}, SOS methods are applied to analyze discrete-time polynomial systems w.r.t. dissipativity properties using noisy data.
Similarly, the recent papers~\cite{guo:depersis:tesi:2020a,guo:depersis:tesi:2020b} extend the results in~\cite{depersis:tesi:2019,waarde:camlibel:mesbahi:2020} to design stabilizing controllers for continuous-time polynomial systems based on noisy data.
In~\cite{dai:sznaier:2021}, a similar control problem is solved using different technical arguments based on Rantzer's Dual Lyapunov approach.
Another important class of nonlinear systems is that of systems with rational dynamics, comprising, e.g., enzyme kinetics~\cite{holmberg:1982} or biochemical reactors~\cite{strogatz:2018}. 
In general, identification for such systems can be challenging~\cite{evans:chapman:chappell:godfrey:2002}.

In the present paper, we consider data-driven controller design for nonlinear continuous-time systems with possibly non-polynomial system dynamics.
To this end, we first extend the results of~\cite{depersis:tesi:2019,waarde:camlibel:mesbahi:2020,berberich:scherer:allgower:2020,guo:depersis:tesi:2020b} to derive a purely data-driven system parametrization of unknown rational systems.
We then exploit this parametrization for robust controller design, adapting existing robust control techniques based on linear matrix inequalities (LMI) in~\cite{scherer:weiland:2000}.
This leads to an SOS-based controller design procedure with robust stability and performance guarantees for all rational systems consistent with the data.
Furthermore, we show that within the proposed framework we are also able to design stabilizing controllers for nonlinear systems with non-polynomial basis functions by lifting them to an extended state-space with polynomial dynamics.

\pubidadjcol
\emph{Outline:}
The paper is organized as follows. In Section \ref{sec:preliminaries}, we state the problem setting and introduce some required notation for SOS optimization. 
The data-driven representation of the class of rational systems used throughout most of this paper is presented in Section \ref{sec:dd_controller_design}.
Based on this, we develop data-driven controller design procedures for closed-loop stability and performance using S-procedure relaxations and we apply the developed technique to numerical examples.
Section~\ref{sec:general_nonlinearities} extends the results to nonlinear systems with non-polynomial basis functions.
Finally, we conclude the paper in Section~\ref{sec:conclusion}.

\emph{Notation:}
We write $I_p$ for the $p\times p$ identity matrix where we omit the index if the dimension is clear from the context.
For a matrix $A$, we denote by $A^{\perp}$ a matrix spanning the left-kernel of $A$, i.e., $A^\perp A=0$.
If $A$ is symmetric, then we write $A\succeq0$ if $A$ is positive semidefinite.
Matrix blocks which can be inferred from symmetry are denoted by $\star$ and we abbreviate $V^\top UV$ by writing $[\star]^\top UV$.
Further, we write $\lVert x\rVert_2$ for the Euclidean norm of a vector $x$.
Finally, $\otimes$ denotes the Kronecker product.

\section{Preliminaries}\label{sec:preliminaries}
In this section, we introduce the problem setting (Section~\ref{sec:problem_setting}) and we provide required background on SOS optimization (Section~\ref{sec:SOS_optimization}).

\subsection{Problem setting}\label{sec:problem_setting}
Throughout most of the paper, we consider continuous-time systems with rational system dynamics of the form
\begin{equation}\label{eq:rational_system_dynamics}
	\begin{aligned}
		\xdot 
		&= f_r(x) + g_r(x) u \\
		&= \begin{bmatrix}
			\frac{a_1(x)}{d_1(x)} \\ 
			\vdots \\ 
			\frac{a_n(x)}{d_n(x)}
		\end{bmatrix}
		+ \begin{bmatrix}
			\frac{b_{11}(x)}{e_{11}(x)} & \cdots & \frac{b_{1m}(x)}{e_{1m}(x)} \\ 
			\vdots & \ddots & \vdots \\ 
			\frac{b_{n1}(x)}{e_{n1}(x)} & \cdots & \frac{b_{nm}(x)}{e_{nm}(x)}
		\end{bmatrix} u
	\end{aligned}
\end{equation}
with the state vector $x(t)\in\bbR^n$, its derivative $\xdot(t)\in\bbR^n$, and the control input $u(t)\in\bbR^m$, all at time $t\geq 0$, where we omit the time index to keep the notation simple. 
Further, $d_i(x)$, $b_{ij}(x)$, $e_{ij}(x)$, $i=1,\dots,n$ and $j=1,\dots,m$, are polynomials with degree greater than or equal to zero, and $a_i(x)$, $i=1,\dots,n$, are polynomials with degree greater than zero. 
The latter is necessary because, for simplicity, we assume the origin to be a steady-state.
Any $\dot{x}$, $x$, $u$ satisfying~\eqref{eq:rational_system_dynamics} also fulfill the polynomial equation
\begin{multline}\label{eq:rational_system_dynamics2}
	p(x) \xdot
	= \begin{bmatrix}
		a_1(x) p_1(x) \\ \vdots \\ a_n(x) p_n(x)
	\end{bmatrix}
	\\
	+ \begin{bmatrix}
		b_{11}(x)p_{11}(x) & \cdots & b_{1m}(x)p_{1m}(x) \\
		\vdots & \ddots & \vdots \\
		b_{n1}(x)p_{n1}(x) & \cdots & b_{nm}(x)p_{nm}(x)
	\end{bmatrix} u
\end{multline}
with
\begin{gather*}
	p(x) = \prod_{i=1}^nd_i(x)\prod_{j=1}^m e_{ij}(x),\\
	p_i(x) = \frac{p(x)}{d_i(x)}, \qquad
	p_{ij}(x) = \frac{p(x)}{e_{ij}(x)}\,.
\end{gather*}
We assume $p(x)\neq 0$ for all $x\in\bbR^n$ which is equivalent to assuming $d_i(x),e_{ij}(x)\neq 0$ for all $i,j$, and thus, implies that the vector field in \eqref{eq:rational_system_dynamics} is globally defined. Then, we define 
\begin{gather*}
	\bmatP \begin{bmatrix}
		Z_p(x) \\ 1
	\end{bmatrix} = p(x),
	\quad
	\cA Z(x) = \begin{bmatrix}
		a_1(x) p_1(x) \\ 
		\vdots \\ 
		a_n(x) p_n(x)
	\end{bmatrix},
	\\
	\cB H(x) = \begin{bmatrix}
		b_{11}(x)p_{11}(x) & \cdots & b_{1m}(x)p_{1m}(x) \\
		\vdots & \ddots & \vdots \\
		b_{n1}(x)p_{n1}(x) & \cdots & b_{nm}(x)p_{nm}(x)
	\end{bmatrix},	
\end{gather*}
where $\cA\in\bbR^{n\times N_z}$, $\cB\in\bbR^{n\times N_u}$ and $\cP\in\bbR^{1\times N_p}$ are unknown parameters representing the coefficients of the polynomials in~\eqref{eq:rational_system_dynamics2}, $Z(x)$ is an $N_z\times 1$ vector of monomials in $x$, $Z_p(x)$ is an $N_p \times 1$ vector of monomials in $x$, and $H(x)$ is an $N_u \times m$ matrix of monomials in $x$.
For instance, in case that the system dynamics~\eqref{eq:rational_system_dynamics} are of a scalar and second-order polynomial form, then we have $Z(x)=H(x)=\begin{bmatrix}x& x^2\end{bmatrix}^\top$ and $Z_p(x)=\emptyset$. 
Throughout this paper, we assume that $Z(x)$, $Z_p(x)$, and $H(x)$ are known, although our arguments remain valid if they are over-approximated, i.e., if additional monomials are added which are not present in the actual unknown system dynamics.
This translates into assuming that the basis functions of the numerators and denominators in~\eqref{eq:rational_system_dynamics} are known or over-approximated.
Moreover, $Z(x)$ and $Z_p(x)$ only contain monomials with a minimum degree of one such that $Z(x)$ is zero if and only if $x=0$ and we can decompose $Z(x)$ as $Z(x)=Y(x)x$ for a matrix $Y(x)$ of dimension $N_z \times n$, and similarly for $Z_p(x)$.

Using these definitions, we can rewrite~\eqref{eq:rational_system_dynamics2} as \begin{equation}\label{eq:rational_system_dynamics3}
	\begin{bmatrix}\cP&1\end{bmatrix}\begin{bmatrix}Z_p(x)\\1\end{bmatrix}
	\xdot = \cA Z(x) + \cB H(x) u\,,
\end{equation}
where $\cP$, $\cA$, and $\cB$ are unknown parameters and $Z_p(x)$, $Z(x)$, and $H(x)$ are known polynomial basis matrices.
This representation is advantageous over~\eqref{eq:rational_system_dynamics} because it is linearly parametrized in the unknown variables.
In this paper, we present an approach for designing polynomial state-feedback controllers $u(x)=K(x)Z(x)$ for the considered rational System \eqref{eq:rational_system_dynamics} with robust stability guarantees, using no model knowledge but only one open-loop data trajectory.
However, instead of assuming that the data are generated exactly by the rational system dynamics~\eqref{eq:rational_system_dynamics}, i.e., they satisfy~\eqref{eq:rational_system_dynamics3}, we allow for a perturbation of~\eqref{eq:rational_system_dynamics3} taking the form
\begin{equation}\label{eq:rational_system_dynamics4}
	\begin{bmatrix}\cP&1\end{bmatrix}\begin{bmatrix}Z_p(x)\\1\end{bmatrix}
	\xdot = \cA Z(x) + \cB H(x) u+B_ww\,,
\end{equation}
where $w(t)\in\mathbb{R}^{m_w}$, $t\geq0$, describes an unknown disturbance sequence perturbing the polynomial equation~\eqref{eq:rational_system_dynamics3}.
We assume that the matrix $B_w$ is known and has full column rank.
If $B_w$ does not satisfy the rank assumption, the disturbance can be transformed into $\tilde{w} = B_w w$ with $\tilde{B}_w = I$ and a quadratic noise bound on the sequence of $\tilde{w}$.
In this paper, we use a finite input-state-trajectory $\{x(t_k),\xdot(t_k),u(t_k)\}_{k=1}^{N}$ with sampling times $\{t_k\}_{k=1}^{N}$ satisfying~\eqref{eq:rational_system_dynamics4} for some unknown disturbance $\{\what(t_k)\}_{k=1}^{N}$ which satisfies a known quadratic bound defined via the matrix $
\hat{W} = \begin{bmatrix}
		\what(t_1) & \what(t_2) & \cdots & \what(t_{N})
\end{bmatrix}
$. 

\begin{assumption}\label{ass:noise_bound}
	The noise generating the data satisfies $\hat{W}\in\cW$, where
	\begin{equation}\label{eq:ass_noise_bound}
		\cW \coloneqq \left\{W\in\bbR^{m_w\times N} \mathrel{\Big|} \begin{bmatrix}
			W^\top\\I
		\end{bmatrix}^\top
		\begin{bmatrix}
			Q_w & S_w \\ S_w^\top & R_w
		\end{bmatrix}
		\begin{bmatrix}
			W^\top\\I
		\end{bmatrix}\succeq 0 \right\}
	\end{equation}
	for known matrices $Q_w\in\bbR^{N\times N}$, $S_w\in\bbR^{N\times m_w}$, and $R_w\in\bbR^{m_w\times m_w}$, where $Q_w \prec 0$.
\end{assumption}
This is a common assumption in the literature, e.g.,~\cite{depersis:tesi:2019,waarde:camlibel:mesbahi:2020,berberich:scherer:allgower:2020} use the same noise bound or a special case thereof for data-driven control of linear systems. With Assumption~\ref{ass:noise_bound},
we require that the unknown noise sequence affecting the measurements lies within a known set $\cW$ and thus, we obtain a quadratic bound on the matrix $\hat{W}$. This assumption includes various relevant scenarios such as, e.g., bounds on the maximal singular value $\sigma_\mathrm{max}(\hat{W})\leq \wbar_\sigma$, or norm bounds, $\|\what(t_k)\|_2\leq \wbar$ for $k=1,\dots,N$. The latter leads to the choice $Q_w = -I$, $S_w = 0$, and $R_w=\wbar^2NI$.

We note that, in general, requiring measurements of the state derivative $\{\xdot(t_k)\}_{k=1}^N$ as above can be restrictive.
However, the values do not need to be known exactly since inaccuracies, possibly resulting from a finite-difference estimation step, can be translated into a disturbance as in~\eqref{eq:rational_system_dynamics4}.
More precisely, if $\xdot$ is affected by bounded measurement noise, then also the corresponding disturbance in~\eqref{eq:rational_system_dynamics4} is bounded with the resulting bound depending on $\{p(x(t_k))\}_{k=1}^N$, i.e., a guaranteed bound can be computed if an upper bound on $p(x(t_k))$ is available.

Further, in the considered problem setting, the noise affecting the measured data does not enter the rational system dynamics~\eqref{eq:rational_system_dynamics} directly, but rather the polynomial equation~\eqref{eq:rational_system_dynamics4}.
The proposed approach can handle disturbances entering the rational system dynamics \eqref{eq:rational_system_dynamics} directly if a bound as in Assumption~\ref{ass:noise_bound} is available for the \emph{transformed} disturbance $p(x)w$.
Such a bound always exists if the measured data are finite, but, as above, it can only be computed in practice if an upper bound on $p(x)$ is available.
Extending the presented results to different noise scenarios is challenging and an interesting issue for future research.

Finally, we note that assuming availability of input-state measurements as done above and, e.g., in~\cite{guo:depersis:tesi:2020a,guo:depersis:tesi:2020b} can be restrictive.
We expect an extension of our results to output-feedback design based on noisy input-output data to be straightforward by using an extended state vector containing the first $n$ derivatives of the input and output, cf. also~\cite{depersis:tesi:2019,berberich:scherer:allgower:2020}.

\subsection{SOS optimization}\label{sec:SOS_optimization}
For a vectorial index $\alpha\in\bbN_0^n$, we write $|\alpha|=\alpha_1 + \cdots + \alpha_n$. Then, we define for a vector $x\in\bbR^n$ the monomial $
x^\alpha = x_1^{\alpha_1}\cdots x_n^{\alpha_n}
$ and let $\bbR[x]$ denote the set of all polynomials $s(x)$ in the variable $x$ with real coefficients, i.e., $
	s(x) = \sum_{\alpha\in\bbN_0^n,|\alpha|\leq d}s_\alpha x^\alpha
$ with $s_\alpha\in\bbR$ for $|\alpha|\leq d\in\bbN_0$. The degree of the polynomial $s(x)$ is defined as the largest $d$ such that $s_\alpha\neq 0$ for some $\alpha\in\bbN_0^n$ with $|\alpha|=d$. Moreover, we denote $\bbR[x]^{p\times q}$ as the set of all $p\times q$-matrices with elements in $\bbR[x]$. The degree of a polynomial matrix is defined as the largest degree of an element of this matrix. We collect all monomials $x^\alpha$ for $|\alpha|\leq d$ in the polynomial vector $z_d(x)$ of length $l(n,d)\coloneqq \binom{n+d}{d}$ with $
	z_d(x) = \begin{bmatrix}
		1 & x_1 & \cdots & x_n & x_1^2 & x_1x_2 & \cdots & x_n^d
	\end{bmatrix}^\top \in \bbR[x]^{l(n,d)}
$.
\begin{definition}[SOS polynomial matrix]
	A polynomial matrix $S(x)\in\bbR[x]^{p\times p}$ is said to be an SOS matrix if there exists a matrix $T(x)\in\bbR[x]^{q\times p}$ with $q\in\bbN$ such that $
		S(x) = T(x)^\top T(x)
	$. For $p=1$, $S(x)$ is called an SOS polynomial.
\end{definition}
Verifying nonnegativity of a polynomial matrix is difficult in general. Since the SOS property implies nonnegativity of the polynomial matrix, SOS matrices are especially interesting from a computational perspective as we can verify the SOS property via an LMI feasibility condition. This is also known as the Gram matrix method \cite{choi:lam:reznick:1994}.
\begin{proposition}[SOS decomposition]
	Let a polynomial matrix $S(x)\in\bbR[x]^{p\times p}$ have degree $2d$. Then, $S(x)$ is an SOS matrix if and only if there exists a real matrix $\Lambda=\Lambda^\top\succeq 0$ such that 
	$
		S(x) = [z_d(x)\otimes I_p]^\top \Lambda [z_d(x)\otimes I_p]
	$.
\end{proposition}
A detailed proof of this proposition can be found in  \cite{chesi:garulli:tesi:vicino:2009}. It is obvious that the characterization of $S(x)$ being SOS is an affine constraint on the matrix $\Lambda$. Verifying the SOS property hence amounts to solving an LMI feasibility problem. Thus, SOS methods provide a computational tool to guarantee global nonnegativity of $S(x)\in\bbR[x]^{p\times p}$, i.e., if $S(x)$ is SOS then $S(x)\succeq 0$ for all $x\in\bbR^n$. 

\section{Data-driven controller design for rational systems}\label{sec:dd_controller_design}
In this section, we design polynomial state-feedback controllers for nonlinear systems with rational dynamics based only on measured data.
In Section~\ref{sec:dd_system_parametrization}, we first provide a data-driven system parametrization based on measured data. Next, in Section~\ref{sec:dd_controller_design_stabilization}, we provide SOS-based design procedures with robust closed-loop stability and performance, which we apply to numerical examples in Section~\ref{sec:dd_controller_design_examples}.

\subsection{Data-driven system parametrization}\label{sec:dd_system_parametrization}
In the following, we provide a simple data-driven parametrization of all rational Systems \eqref{eq:rational_system_dynamics} which are consistent with the measured data $\{x(t_k),\xdot(t_k),u(t_k)\}_{k=1}^{N}$ and with the noise bound $\hat{W}\in\cW$.

We denote the set of all open-loop matrices $A$, $B$, and $P$ which are consistent with the data and the noise bound by
\begin{multline*}
	\SigABP \coloneqq 
	\left\{
		\bmatABP \mathrel{\Big|} Z_p(P, X, \Xdot)
	\right. \\ 
	\left.
		= AZ(X) + BH(X,U) + B_wW,\,W\in\cW
	\right\},
\end{multline*}
where $X = \begin{bmatrix}
	x(t_1) & \cdots & x(t_{N})
\end{bmatrix}$, $U$, $\Xdot$, $Z(X)$ are defined analogously, and
\begin{align*}	
	Z_p(P,X,\Xdot) &= (I_n \kron P) \hZ_p(X,\Xdot) + \Xdot, \\
	\hZ_p(X,\Xdot) &= \begin{bmatrix}
		\tZ_p(x(t_1))\xdot(t_1) & \cdots & \tZ_p(x(t_{N}))\xdot(t_N)
	\end{bmatrix},\\
	H(X,U) &= \begin{bmatrix}
		H(x(t_1))u(t_1) & \cdots & H(x(t_{N}))u(t_{N})
	\end{bmatrix},
\end{align*}
with $\tZ_p(x) = I_n \kron Z_p(x)$.
Note that the set $\SigABP$ also contains the \emph{true} matrices $\cA$, $\cB$ and $\cP$ of the system, i.e., $\begin{bmatrix}
	\cA & \cB & \cP
\end{bmatrix}\in\SigABP$. Moreover, we define the set 
{\small
	\begin{equation*}
		\cM \coloneqq \left\{\bmatABP \mathrel{\Bigg|} 
		\left[\star\right]^\top
		\begin{bmatrix}
			\oQ_w & \oS_w \\ \oS_w^\top & \oR_w
		\end{bmatrix}
		\begin{bmatrix}
			\begin{bmatrix}
				\bmatAB^\top \\
				(I_n \kron P)^\top
			\end{bmatrix} \\
				I
		\end{bmatrix}
		\succeq 0 \right\},
	\end{equation*}
}
where 
\begin{equation*}
	\begin{bmatrix}
		\oQ_w & \oS_w \\ \oS_w^\top & \oR_w
	\end{bmatrix} = 
	\begin{bmatrix}
		\begin{bmatrix}
			-\begin{bmatrix}
				Z(X)\\H(X,U)
			\end{bmatrix}\\\hZ_p(X, \Xdot)
		\end{bmatrix} & 0 \\ 
		\Xdot & B_w
	\end{bmatrix}
	\begin{bmatrix}
		Q_w & S_w \\ S_w^\top & R_w
	\end{bmatrix}
	\left[\star\right]^\top,
\end{equation*}
which depends on the noise bound and on the measured data through $\Xdot$, $Z(X)$, $H(X,U)$, $\hZ_p(X,\Xdot)$. Then, the following theorem establishes that $\cM$ is an equivalent parametrization of $\SigABP$.

\begin{theorem}\label{thm:set_equivalence}
	Suppose Assumption \ref{ass:noise_bound} is satisfied. Then, it holds that $\SigABP=\cM$.
\end{theorem}

\begin{proof}
	Note that $\bmatABP\in\cM$ if and only if 
	{
	\begin{equation}\label{eq:cM}
	\begin{bmatrix}F^\top\\B_w^\top	\end{bmatrix}^\top
			\begin{bmatrix}
				Q_w & S_w \\ S_w^\top & R_w
			\end{bmatrix}
			\begin{bmatrix}F^\top\\B_w^\top	\end{bmatrix}
			\succeq 0\,,
		\end{equation}
	}
	where we abbreviate
	\begin{align*}
	F=Z_p(P,X,\Xdot) - A Z(X) - B H(X,U).
	\end{align*}

	(i) \textbf{Proof of $\SigABP\subseteq \cM$:}
		Suppose that $\bmatABP\in\SigABP$, i.e., there exists $W\in\cW$ such that $Z_p(P,X,\Xdot) = AZ(X) + BH(X,U) + B_wW$. Then, $W\in\cW$ together with \eqref{eq:ass_noise_bound} implies with multiplication from the left and right by $B_w$ and its transpose, respectively, 
		\begin{equation*}
			\begin{bmatrix}
				(B_wW)^\top \\ 
				B_w^\top
			\end{bmatrix}^\top 
			\begin{bmatrix}
				Q_w & S_w \\ S_w^\top & R_w
			\end{bmatrix}
			\begin{bmatrix}
					(B_wW)^\top \\ 
					B_w^\top
			\end{bmatrix} 
			\succeq 0\,.
		\end{equation*} 
		Replacing $
			B_wW 
			= Z_p(P, X, \Xdot)
			- A Z(X) - B H(X,U)
		$, we deduce that $\bmatABP$ satisfies Inequality \eqref{eq:cM} and hence, $\bmatABP\in\cM$.
			
	(ii) \textbf{Proof of $\cM\subseteq\SigABP$:}
		Let $\bmatABP\in\cM$, i.e., $\bmatABP$ satisfies Inequality \eqref{eq:cM}. Multiplying Inequality \eqref{eq:cM} from the left and right by $B_w^\perp$ and its transpose, respectively, we obtain
		{\small
			\begin{equation*}
				\left(
					B_w^\perp
					\big(
						Z_p(P,X,\Xdot) 
						- A Z(X) - B H(X,U)
					\big) 
				\right)
				Q_w 
				\left[\star\right]^\top
				\succeq 0\,.
			\end{equation*}
		}
		Since $Q_w \prec 0$, this implies $
			B_w^\perp
			\big(
				Z_p(P,X,\Xdot) 
				- AZ(X) - BH(X,U)
			\big) = 0
		$, i.e., there exists $W$ such that $
			Z_p(P,X,\Xdot) 
			- AZ(X) - BH(X,U)
			= B_w W
		$. Substituting this into Inequality \eqref{eq:cM} yields
		\begin{equation*}
			\begin{bmatrix}
				(B_wW)^\top \\ 
				B_w^\top
			\end{bmatrix}^\top
			\begin{bmatrix}
				Q_w & S_w \\ S_w^\top & R_w
			\end{bmatrix}
			\begin{bmatrix}
				(B_wW)^\top \\ 
				B_w^\top
			\end{bmatrix} 
			\succeq 0\,.
		\end{equation*}
		Since $B_w$ has full column rank, this implies $W\in\cW$ and hence, $\bmatABP\in\SigABP$.
\end{proof}

Via Theorem \ref{thm:set_equivalence}, we can formulate a \emph{single} quadratic matrix inequality equivalently describing all systems consistent with the data, i.e., satisfying the data equation $Z_p(P,X,\Xdot) = AZ(X) + BH(X,U) + B_w W$ for some $W\in\cW$.
The key idea is to choose basis matrices $Z(x)$, $Z_p(x)$, and $H(x)$ such that the resulting polynomial system description \eqref{eq:rational_system_dynamics4} is linear in the unknown parameters and hence, we can follow a similar approach as in the corresponding results for linear systems (compare \cite{waarde:camlibel:mesbahi:2020,berberich:scherer:allgower:2020}). 
We note that an analogous approach is used for data-driven control of polynomial systems in~\cite{guo:depersis:tesi:2020b} and, in particular, Theorem~\ref{thm:set_equivalence} reduces to the corresponding result in~\cite{guo:depersis:tesi:2020b} for polynomial systems.

\subsection{Controller design for robust stability and performance}\label{sec:dd_controller_design_stabilization}
Next, we design polynomial state-feedback controllers for System \eqref{eq:rational_system_dynamics} taking the form $u(x)=K(x)Z_K(x)$ with some polynomial matrix $K(x)\in\bbR[x]^{m\times N_K}$ and a known vector of monomials $Z_K(x)\in\bbR[x]^{N_K}$, $Z_K(x)=Y_K(x)x$, which may differ from $Z(x)$.
To this end, we employ the parametrization provided by Theorem~\ref{thm:set_equivalence} to achieve robust closed-loop stability for all rational systems consistent with the measured data and the noise bound. 
In the following, we abbreviate 
\begin{gather*}
	q_1(x) = \begin{bmatrix}
		\begin{bmatrix}
			Y(x) \\ H(x)K(x)Y_K(x)
		\end{bmatrix} \\ 0
	\end{bmatrix},
\quad
	q_2(x) = \begin{bmatrix}
		0 \\ I_n \kron Z_p(x)
	\end{bmatrix}.
\end{gather*}
\begin{theorem}\label{thm:cl_stability_sos}
	Suppose Assumption \ref{ass:noise_bound} holds. If there exist an $n\times n$ matrix $\cY\succ 0$, a polynomial matrix $K(x)\in\bbR[x]^{m\times N_K}$, and a scalar $\tau\geq0$ such that the matrix polynomial 
	\begin{equation}\label{eq:cl_stability_condition_sos}
		Q(x) = -
		\left[\star\right]^\top
		\left[
		\def\arraystretch{1.15}\begin{array}{cc|cc}
			\eps I & \cY & 0 & 0 \\
			\cY & 0 & 0 & 0 \\\hline 
			0 & 0 & \tau \oQ_w & \tau \oS_w \\[1ex]
			0 & 0 & \tau \oS_w^\top & \tau \oR_w
		\end{array}
		\right]
		\left[
		\def\arraystretch{1.0}\begin{array}{cc}
			I & q_2(x)^\top \\
			0 & q_1(x)^\top \\\hline 
			0 & I \\ 
			I & 0 
		\end{array}
		\right]
	\end{equation}
	is SOS for some $\eps\geq 0$, then the controller $u(x) = K(x) Z_K(x)$ globally stabilizes System \eqref{eq:rational_system_dynamics} for all $\bmatABP\in\SigABP$.
	
	Moreover, if $Q(x)$ is SOS for some $\eps>0$, then the controller globally asymptotically stabilizes System \eqref{eq:rational_system_dynamics} for all $\bmatABP\in\SigABP$.
\end{theorem}
\begin{proof}
	Consider the Lyapunov function candidate $V(x)=x^\top\cX x$, $\cX=\cY^{-1}$ for System \eqref{eq:rational_system_dynamics}. This candidate is clearly positive definite since $\cX\succ 0$. We show $\Vdot(x) \leq -\eps \|\cX x\|_2$, which implies stability with respect to the origin for $\eps \geq 0$ and asymptotic stability for $\eps > 0$. 
	
	Using that $Q(x)\succeq 0$ for all $x\in\bbR^n$, we obtain 
	\begin{equation*}
		-\begin{bmatrix}
			I \\
			\begin{bmatrix}
				\bmatAB^\top \\
				(I_n \kron P)^\top
			\end{bmatrix}
		\end{bmatrix} ^\top
		Q(x)
		\begin{bmatrix}
			I \\
			\begin{bmatrix}
				\bmatAB^\top \\
				(I_n \kron P)^\top
			\end{bmatrix}
		\end{bmatrix} \preceq 0
	\end{equation*}
	for all $x\in\bbR^n$. 
	Then, applying the S-procedure (cf.~\cite{scherer:weiland:2000,boyd:vandenberghe:2004}) yields for all~$x\in\bbR^n$ 
	{\normalsize
		\begin{multline}\label{eq:cl_stability_condition_s_procedure}
			x^\top\cX
			\left(
				(AY(x) + BH(x)K(x)Y_K(x))\cY p(x) + p(x)^2\eps I
			\right.\\
			\left.
				+ p(x)\cY 
				(AY(x) + BH(x)K(x)Y_K(x))^\top
			\right)
			\cX x 
			\leq 0, \\
			\forall \bmatABP\,:\,
			\left[\star\right]^\top
			\begin{bmatrix}
				\oQ_w & \oS_w \\ \oS_w^\top & \oR_w
			\end{bmatrix}
			\begin{bmatrix}
				\begin{bmatrix}
					\bmatAB^\top \\
					(I_n \kron P)^\top
				\end{bmatrix} \\ I
			\end{bmatrix} \succeq 0\,,
		\end{multline}
	}
	where the first inequality is obtained by additionally multiplying with $x^\top\cX$ and its transpose from left and right, respectively, and by using that $(I_n \kron P)(I_n \kron Z_p(x))+I_n=p(x)I_n$.
	Let now $\bmatABP\in\SigABP$. By Theorem~\ref{thm:set_equivalence}, this implies $\bmatABP\in\cM$ and hence, together with \eqref{eq:cl_stability_condition_s_procedure}, $Z(x)=Y(x)x$, and $Z_K(x)=Y_K(x)x$, 
	\begin{multline*}
		\hspace*{-7pt}
		p(x)\left(
			x^\top\cX (AZ(x)+BH(x)K(x)Z_K(x)) + p(x) \eps \|\cX x\|_2 
		\right.\\
		\left.
			+ (AZ(x)+BH(x)K(x)Z_K(x))^\top\cX x
		\right) 
		 \leq 0\,.
	\end{multline*}
	The system dynamics~\eqref{eq:rational_system_dynamics} imply $p(x)\xdot = AZ(x)+BH(x)K(x)Z_K(x)$ such that we obtain the Lyapunov inequality $p(x)^2(x^\top\cX \xdot + \xdot^\top\cX x + \eps \|\cX x\|_2) \leq 0$. 
	Using that $p(x)^2\neq0$ for all $x\in\mathbb{R}^n$, this implies stability for any $\bmatABP\in\SigABP$.
	Asymptotic stability follows for $\eps>0$, which thus concludes the proof.
\end{proof}
Note that the polynomial matrix $Q(x)$ in~\eqref{eq:cl_stability_condition_sos} is not linear in the decision variables $\cY$ and $K(x)$ and hence, for a practical implementation, it needs to be transformed into a linear SOS condition.
This is possible following standard steps from model-based robust control of linear systems (compare~\cite{scherer:weiland:2000}): Defining the new variable $L(x)=K(x)Y_K(x)\cY$, 
%
$Q(x)$ is linear in the decision variables $\cY$, $L(x)$, and $\tau$. After finding variables $\cY$, $L(x)$, and $\tau$ such that $Q(x)$ is SOS, a controller stabilizing the unknown system can be implemented as
\begin{equation}\label{eq:stabilizing_input}
	u(x) = K(x)Z_K(x) = K(x) Y_K(x) x = L(x)\cY^{-1} x\,.
\end{equation}
Hence, Theorem~\ref{thm:cl_stability_sos} provides a simple and direct method to design a controller which robustly stabilizes the rational System~\eqref{eq:rational_system_dynamics} for all matrices which are consistent with the data and the noise bound. Moreover, the number of decision variables of the proposed SOS program is independent of the data length which allows us to consider large data sets for the controller design. 
More precisely, the feasibility problem in Theorem~\ref{thm:cl_stability_sos} has overall $1+\frac{n(n+1)}{2}+mN_K$ decision variables, i.e., the number of decision variables scales quadratically with the system order $n$ as in standard robust control and linearly with the number of monomials in $Z_K(x)$.

Theorem~\ref{thm:cl_stability_sos} reduces to a recently obtained result in~\cite{guo:depersis:tesi:2020b} for polynomial system dynamics, i.e., if $d_i(x)=e_{ij}(x)=1$ for all $i=1,\dots,n$ and $j=1,\dots,m$.
The above result is not a straightforward consequence of the result in~\cite{guo:depersis:tesi:2020b} due to the role of the polynomial $p(x)$ in the proof.
More precisely, only by combining the bound on $\begin{bmatrix}A&B&P\end{bmatrix}$ in Theorem~\ref{thm:set_equivalence} with the fact that a quadratic Lyapunov function is used, we can obtain~\eqref{eq:cl_stability_condition_s_procedure} via the identity
\begin{align*}
(I_n \kron P)(I_n \kron Z_p(x))+I_n=p(x)I_n
\end{align*} 
in the proof.
If, e.g., more general polynomial Lyapunov functions $V(x) = Z(x)^\top\cX(x)Z(x)$ are used, the partial derivatives $\frac{\partial Z(x)}{\partial x}$ prohibit this argument and thus, an extension of our results to polynomial Lyapunov functions is an interesting issue for future research.

An additional advantage of our approach is that the above arguments follow the line of LMI-based robust control techniques from~\cite{scherer:weiland:2000} and it is straightforward to extend them to more general problem formulations such as including prior knowledge on the system (cf.~\cite{berberich:scherer:allgower:2020}) or designing a controller which not only stabilizes the closed loop but also enforces a desired performance specification.
To this end, we consider robust quadratic performance (cf. \cite{scherer:weiland:2000}) with index $\begin{bmatrix}
	Q_p & S_p \\ S_p^\top & R_p
\end{bmatrix}$, $R_p\succ 0$, for the system $\xdot = f_r(x) + g_r(x)u + B_p w_p$ and performance output $z_p = C Z(x) + Du + D_p w_p$. 
That is, for any trajectory of this system there exists $\delta >0$ such that
\begin{equation*}
	\int_0^{\infty} 
	\left[\star\right]^\top 
	\begin{bmatrix}
		Q_p & S_p \\ S_p^\top & R_p
	\end{bmatrix}
	\begin{bmatrix}
		w_p(t) \\ z_p(t)
	\end{bmatrix}
	\dd t
	\leq -\delta^2 \int_0^\infty \|w_p(t)\|_2^2 \dd t\,.
\end{equation*}
For instance, $Q_p=-\gamma^2 I$, $S_p=0$, $R_p=I$ corresponds to an $\cL_2$-gain bound $\gamma$ on $w_p \mapsto z_p$.
Following similar arguments as in Theorem \ref{thm:cl_stability_sos}, it can be shown that robust quadratic performance holds for the channel $w_p \mapsto z_p$ if there exist variables $\cY$, $L(x)$ and $\tau$ such that
\begin{equation}\label{eq:performance_sos}
	\left[
		\def\arraystretch{1.15}\begin{array}{cc|cc}
			\multicolumn{2}{c|}{\multirow{2}{*}{$Q(x)$}}
			& \star
			& \star
			\\
			&
			& \star
			& \star
			\\\hline
			-B_p^\top - S_p q_3(x)
			& -B_p^\top q_2(x)^\top
			& -\cQ_p
			& \star 
			\\
			q_3(x)
			& q_3(x)q_2(x)^\top 
			& D_p 
			& R_p^{-1}
		\end{array}
	\right]
\end{equation} 
is SOS, where $q_3(x) = CY(x)\cY + DL(x)$ and $\cQ_p = Q_p + S_p D_p + D_p^\top S_p^\top$. 
If this problem is feasible, then the controller in \eqref{eq:stabilizing_input} guarantees quadratic performance robustly for all rational systems which are consistent with the measured data and the noise level. The proof of this fact is omitted due to space reasons.

We note that the authors in \cite{martin:allgower:2020} use related tools to analyze \emph{discrete-time} polynomial systems w.r.t. dissipativity properties but an extension to controller design in discrete-time is challenging. 

\subsection{Numerical example}\label{sec:dd_controller_design_examples}
Finally, we illustrate the applicability of the proposed approach with two numerical examples: an academic example as well as a realistic model used for drug distribution \cite{chappell:godfrey:vajda:1990}. For the latter, we also incorporate performance specifications.
The simulations in this paper are implemented in \textsc{Matlab} using YALMIP~\cite{lofberg:2004} with its SOS module~\cite{lofberg:2009} and the solver MOSEK~\cite{mosek:2020}.

\begin{example}\label{ex_1}
	Consider the rational system 
	\begin{equation}\label{eq:exmp_system}
		\begin{aligned}
			\xdot_1 &= \frac{x_2^2}{1+x_1^2} + u_1\,,\\
			\xdot_2 &= x_1 x_2 + x_2 u_2\,.
		\end{aligned}
	\end{equation}
	According to \eqref{eq:rational_system_dynamics2}, any trajectory of this system also satisfies 
	\begin{equation*}
		\begin{aligned}
			(1+x_1^2)\xdot_1 &= x_2^2 + (1 + x_1^2)u_1\,,\\
			(1+x_1^2)\xdot_2 &= x_1 x_2 + x_1^3 x_2 + (x_2 + x_1^2x_2) u_2 \,.
		\end{aligned}
	\end{equation*} 
	Following Section~\ref{sec:problem_setting}, we define the true (unknown) system matrices $
		\cA = \begin{bmatrix}
			0 & 0 & 1 & 0 & 0 \\
			0 & 0 & 0 & 1 & 1
		\end{bmatrix}
	$, $
		\cB = \begin{bmatrix}
			1 & 1 & 0 & 0 \\
			0 & 0 & 1 & 1
		\end{bmatrix}
	$, and $
		\cP = 1
	$, and the vectors of monomials $
		Z(x) = \begin{bmatrix}
			x_1 & x_2 & x_2^2 & x_1 x_2 & x_1^3 x_2
		\end{bmatrix}^\top
	$, $
		Z_p(x) = x_1^2
	$, and $
		H(x) = \begin{bmatrix}
				1 & x_1^2 & 0 & 0 \\
				0 & 0 & x_2 & x_1^2 x_2
		\end{bmatrix}^\top
	$. 
	For the data generation, we consider $B_w=I$ and noise sampled uniformly from the ellipse $\|\what(t_k)\|_2 \leq \wbar$, $k=1,\dots,N$, leading to the choices $Q_w=-I$, $S_w=0$, and $R_w=\wbar^2NI$ as introduced in Section~\ref{sec:problem_setting}. 
	
	\begin{table*}[!t]
		\vspace{4pt}
		\centering
		\caption{Number of feasible designs for Example \ref{ex_1}.}
		\label{tab:exmp_1}
		\vspace*{-5pt}
		\begin{tabular}{c|ccccccc}
			$N$ & $\wbar = \num{1e-6}$ & $\wbar = \num{1e-4}$ & $\wbar = \num{1e-3}$ & $\wbar = \num{2e-3}$ & $\wbar=\num{5e-3}$ & $\wbar = \num{1e-2}$ & $\wbar = \num{1e-1}$ \\\hline
			$20$	& $20$ 	& $19$ 	& $0$ 	& $0$ 	& $0$ 	& $0$ 	& $0$ \\
			$100$	& $19$ 	& $15$ 	& $1$ 	& $0$ 	& $0$ 	& $1$ 	& $0$ \\
			$1000$	& $20$ 	& $19$ 	& $14$ 	& $13$ 	& $0$ 	& $0$ 	& $0$ \\
			$10000$ & $20$ 	& $20$	& $20$ 	& $17$ 	& $20$ 	& $19$ 	& $9$ \\
			$20000$ & $19$ 	& $20$	& $20$	& $18$ 	& $18$	& $17$ 	& $19$
		\end{tabular}
		\vspace*{-15pt}
	\end{table*}
	Since the considered system is open-loop unstable, we generate the data via concatenation of multiple trajectories of length $N_d=5$. More precisely, we measure $d=N/N_d$ trajectories with data matrices $\{\Xdot_1,X_1,U_1\},\ldots,\{\Xdot_d,X_d,U_d\}$, which are affected by noise matrices $\hat{W}_1,\ldots,\hat{W}_d$, respectively, and we define $
	\Xdot = \begin{bmatrix}
		\Xdot_1 & \cdots & \Xdot_d
	\end{bmatrix}
	$, $
	X = \begin{bmatrix}
		X_1 & \cdots & X_d
	\end{bmatrix}
	$, $
	U = \begin{bmatrix}
		U_1 & \cdots & U_d
	\end{bmatrix}
	$, and consider a noise bound of the form $
	\begin{bmatrix}
		\hat{W}_1 & \cdots & \hat{W}_d
	\end{bmatrix}\in\cW
	$. We use equidistant data points generated by the true system dynamics with sampling rate $T_s = 0.001$ and initial state drawn uniformly from the interval $[-1,1]^2$. At each time step, we inject an input sampled uniformly from $[-5,5]^2$.

	In the following, we analyze the feasibility of Theorem~\ref{thm:cl_stability_sos} for different data lengths and noise bounds with $Z_K(x)=Z(x)$.
	Since Theorem~\ref{thm:cl_stability_sos} provides conditions for \emph{robust} stabilization of all systems consistent with the data and hence, only sufficient conditions for stabilization of~\eqref{eq:exmp_system}, the feasibility can vary when considering different noise instances satisfying the same bound. 
	Therefore, for each combination of noise level $\wbar$ and data length $N$, we perform $20$ experiments and record how often the resulting SOS condition on $Q(x)$ is feasible with $\eps=\num{1e-7}$ and leads to an asymptotically stabilizing controller. 
	Table~\ref{tab:exmp_1} displays the number of successful designs of the controller according to our requirements.
	First, note that the approach requires the data to be sufficiently rich, i.e., it does not lead to a stabilizing controller if the data length is too small. 
	The definition of $\oQ_w$ requires at least $N \geq N_z + N_u + nN_p=11$ samples but due to the noise, even more samples are required. 
	Generally, the feasibility of the conditions in Theorem~\ref{thm:cl_stability_sos} improves for increasing data lengths and deteriorates for increasing noise levels. 
	The improvement in the data length is not strictly monotone since the noise bound in Assumption~\ref{ass:noise_bound} does not exactly capture the actual pointwise-in-time noise bound.
	Interestingly, even for the considered \emph{two-dimensional} example, conservatism is reduced and robustness is increased if the number of data points is as large as $N=20000$.
	Finally, we note that the number of decision variables of the proposed approach is independent of the data length and hence, the SOS problem in~\eqref{eq:cl_stability_condition_sos} can be solved in less than fifteen seconds on a standard Intel Core i7 notebook even for $N=20000$. 
\end{example}

\begin{example}\label{ex_2}
	To demonstrate the practicality of our approach, we also apply it to a two-compartment model used to describe a drug distribution \cite{chappell:godfrey:vajda:1990}, i.e., the system
	\begin{equation}\label{eq:exmp_system_2}
		\begin{aligned}
			\xdot_1 &= -\frac{x_1}{5+x_1} - x_1 + x_2 + u\,,\\
			\xdot_2 &= x_1 - x_2\,.
		\end{aligned}
	\end{equation}
	Note that this system violates our standing assumption $p(x)\neq0$ for all $x\in\mathbb{R}^n$.
	Nevertheless, as we show in the following, the presented approach can be used to successfully design a robust controller with local closed-loop stability and good practical performance.
	For the data generation, we proceed as for Example~\ref{ex_1} with $\wbar=\num{1e-1}$ and $d=40$ trajectories of length $N_d=5$ sampled uniformly from $x(t_0)\in[-2,2]^2$, $u(t_k)\in[-5,5]$ and $\|\what(t_k)\|_2 \leq \wbar$ for $k=1,\dots,N_d$. 
	
	We choose $Q_p=-\gamma^2 I$, $S_p=0$, $R_p=I$, $B_p=I$, and $z_p=x$ such that our design objective is an $\cL_2$-gain bound of $\gamma=400$ on the channel $w_p\mapsto x$.
	Then, we solve \eqref{eq:performance_sos} to obtain a controller. 
	The phase portrait of the resulting closed-loop behavior with the controller designed using \eqref{eq:performance_sos} is depicted in Fig.~\ref{fig:exmp_closed_loop}. It shows that the controller stabilizes the system at the origin for $x\in(-5,5)\times(-5,15)$ even though $p(x)\neq 0$ is not satisfied globally.	
	\begin{figure}[!t]
		\centering
		\includegraphics[width=\linewidth]{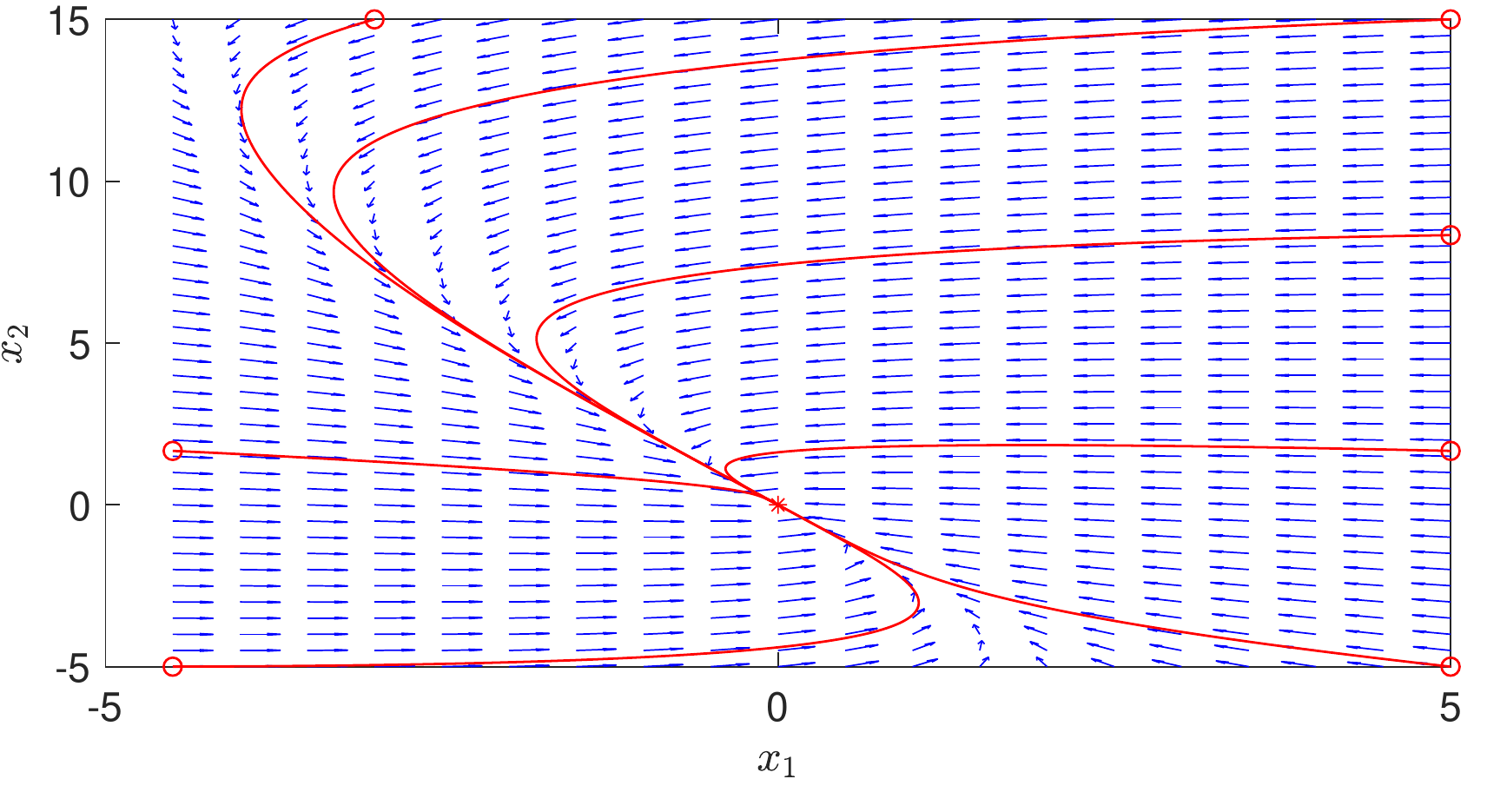}
		\vspace*{-22pt}
		\caption{Vector plot of the closed-loop system in Example \ref{ex_2} with trajectories starting at different initial conditions ({\color{red}$\circ$}) converging to the origin ({\color{red}$*$}).}
		\label{fig:exmp_closed_loop}
		\vspace*{-16pt}
	\end{figure}
	To illustrate the benefits of the performance criterion, we also compute a stabilizing controller based on \eqref{eq:cl_stability_condition_sos}. Compared to this controller, the controller computed via \eqref{eq:performance_sos} leads to a $20\%$ reduction of the $\cL_2$-norm of the performance output $z_p$.
	For this example, \eqref
	{eq:performance_sos} is solvable in less than one second.
\end{example} 

\section{Beyond polynomial basis functions}\label{sec:general_nonlinearities}
In this section, we propose a data-driven control approach for systems which are linear in general, possibly non-polynomial basis functions. 
The main idea relies on lifting the nonlinear system to an extended state-space with polynomial dynamics and then applying the results of Section~\ref{sec:dd_controller_design} for robust controller design.

In the following, we consider nonlinear systems 
\begin{equation}\label{eq:nonlinear_system_dynamics}
    \xidot = f(\xi) + g(\xi)u,
\end{equation}
where the nonlinear functions $f\func,g\func$ can be written as linear combinations of known basis functions $\{\psi_i(\xi)\}_{i=1}^L$. 
Under suitable assumptions on these basis functions, \eqref{eq:nonlinear_system_dynamics} can be lifted to a polynomial system in an extended state-space.
The idea was presented in \cite{gu:2011} as part of a model order reduction method and was recently applied in the context of system identification \cite{qian:kramer:peherstorfer:willcox:2020} and extended dynamic mode decomposition \cite{netto:susuki:krishnan:zhang:2020}. To the best of our knowledge, this work is the first which uses this idea for data-driven control.

The main concept is to exploit the invariance of the chosen basis functions, i.e., the fact that the derivatives of the basis functions can be described solely in terms of the basis functions themselves. To this end, we include each basis function as an additional state coordinate and thus, we obtain the dynamics of the extended state by building the Lie derivative of each basis function w.r.t. \eqref{eq:exmp_nonlinear_dynamics}. Due to the above-described invariance property of the basis functions, we can substitute each nonlinearity, i.e., composition of the basis functions, by the respective state coordinates corresponding to the basis functions.
This so-called \emph{polynomialization} leads to a system representation of \eqref{eq:nonlinear_system_dynamics} which has a larger state dimension but can be described by a polynomial function linear in $u$ \cite[Thm. 1]{gu:2011}. 

Since $f\func$ and $g\func$ are assumed to be linear in the basis functions, we find $\{\psi_i(\xi)\}_{i=1}^L$ and define $x=\Psi(\xi)=\begin{bmatrix}
	\psi_1(\xi) & \cdots & \psi_L(\xi)
\end{bmatrix}^\top$ such that System \eqref{eq:nonlinear_system_dynamics} can be polynomialized leading to
\begin{equation}\label{eq:nonlinear_system_dynamics2}
	\xdot = A Z(x) + B H(x) u \,,
\end{equation}
where $Z(x)\in\bbR[x]^{N_z}$ and $H(x)\in\bbR[x]^{N_u\times m}$. 
Possible basis functions consist of a composition of suitable \emph{elementary} functions, e.g., the functions listed in Table~\ref{tab:elementary_functions}. 
Note that, due to invariance, some elementary functions can only be chosen in pairs, e.g., since $\cos(x)$ is the derivative of $\sin(x)$.
The approach is not limited to the elementary functions in Table~\ref{tab:elementary_functions}, but qualifies for every elementary function $z(x)$ whose gradient $\dpartial{z(x)}{x}$ is polynomial in $z$. Nonlinear systems linear in such basis functions are common in, e.g., chemical rate equations, circuit simulation and mechanical applications. 
\begin{example}
    Consider the nonlinear system $\xidot = \tanh(\xi) + u$. By introducing $\psi(\xi)=\tanh(\xi)$ and $x=(\xi,\psi(\xi))$, we obtain
    \begin{align*}
        \xdot_1 &= x_2 + u\,, \\
        \xdot_2 &= \dpartial{\psi(\xi)}{\xi}\xidot = (1-\tanh(\xi)^2)(\tanh(\xi) + u) \\
        &= x_2 - x_2^3 + (1-x_2^2)u\,,
    \end{align*}
    which is polynomial in $x$ and linear in $u$. 
\end{example}
In general, the polynomialization is not unique, i.e., by introducing more basis functions we obtain a larger state dimension but a possibly smaller polynomial degree. 
Note that the complexity of the robust control approach presented in Section~\ref{sec:dd_controller_design} scales both with the system dimension and the degree of the polynomials, leading to a trade-off for a suitable choice of the polynomialization.

Polynomialization allows us to consider general nonlinear systems with possibly non-polynomial basis functions in the framework of direct data-driven control of polynomial systems.
Similar to the discussion in Section~\ref{sec:problem_setting}, we allow for a perturbation in the measured data taking the form 
\begin{equation}\label{eq:nonlinear_system_dynamics3}
    \xdot = A Z(x) + B H(x) u + B_w w\,,
\end{equation}
where $w$ describes an unknown disturbance sequence perturbing the polynomial equation \eqref{eq:nonlinear_system_dynamics2}.
Hence, designing a stabilizing controller for the nonlinear system \eqref{eq:nonlinear_system_dynamics} reduces to controller design of the polynomial system \eqref{eq:nonlinear_system_dynamics3}. Since System~\eqref{eq:nonlinear_system_dynamics3} is a special case of System~\eqref{eq:rational_system_dynamics4} with $P=Z_p(x)=\emptyset$ and $p(x)=1$, we can apply Theorem~\ref{thm:cl_stability_sos} for controller design.

\begin{table}[!t]
	\vspace{4pt}
	\centering
	\caption{Elementary functions for polynomialization.}
	\label{tab:elementary_functions}
	\vspace*{-5pt}
	\begin{tabular}{c|cc}
		\hline
		Function & Lie derivative \\\hline
		$z = e^x$ & $\zdot = z \xdot$ \\\hline
		$z = \frac{1}{x + a}$, $a\in\bbR$ & $\zdot = -z^2 \xdot$ \\\hline
		$z_1 = \sin(x)$ & $\zdot_1 = -z_2 \xdot$ \\
		$z_2 = \cos(x)$ & $\zdot_2 = -z_1 \xdot$ \\\hline
		$z_1 = \ln(x)$ & $\zdot_1 = z_2 \xdot$ \\
		$z_2 = x^{-1}$ & $\zdot_2 = -z_2^2 \xdot$ \\\hline
		$z_1 = \sqrt{x}$ & $\zdot_1 = \frac{1}{2} z_2 \xdot$ \\
		$z_2 = \frac{1}{\sqrt{x}}$ & $\zdot_2 = -\frac{1}{2}z_2^3 \xdot$ \\\hline
	\end{tabular}
	\vspace*{-15pt}
\end{table}
\begin{remark}
    Generally, polynomialization introduces some conservatism, i.e., System \eqref{eq:nonlinear_system_dynamics} can be described by \eqref{eq:nonlinear_system_dynamics2} but not vice versa. 
    For equivalence, we would need to explicitly include equality constraints ensuring the dependencies of the basis functions.
    Nevertheless, the presented approaches allows us to design controllers for unknown systems \eqref{eq:nonlinear_system_dynamics} with general nonlinear basis functions using only measured data. 
    As we illustrate with the following example, this allows us to solve practically relevant control problems which cannot be handled using existing approaches.
\end{remark}

\begin{example}\label{sec:general_nonlinearities_example}
	Consider the system dynamics of an undamped pendulum
	\begin{equation}\label{eq:exmp_nonlinear_dynamics}
		\begin{aligned}
			\xidot_1 &= \xi_2 \\
			\xidot_2 &= - \frac{b}{m} \xi_2 - \frac{g}{l} \sin(\xi_1) + u
		\end{aligned}
	\end{equation}
	with $b=0$, $m=1$, $g=9.81$, $l=1$. We define the elementary functions $\psi_1(\xi) = \sin(\xi_1)$ and $\psi_2(\xi)=\cos(\xi_1)$. Then, we can polynomialize \eqref{eq:exmp_nonlinear_dynamics} with the extended state $x=\begin{bmatrix}
		\xi^\top & \sin(\xi_1) & \cos(\xi_1)
	\end{bmatrix}^\top$ leading to the polynomial dynamics
	\begin{equation*}
		\xdot = \begin{bmatrix}
			0 & 1 & 0 & 0 & 0 & 0 \\
			0 & -\frac{b}{m} & -\frac{g}{l} & 0 & 0 & 0\\
			0 & 0 & 0 & 0 & 0 & 1 \\
			0 & 0 & 0 & 0 & -1 & 0
		\end{bmatrix} Z(x) + \begin{bmatrix}
			0 \\ 1 \\ 0 \\ 0
		\end{bmatrix} u,
	\end{equation*}
	where $Z(x)=\begin{bmatrix}
		x^\top & x_2 x_3 & x_2 x_4
	\end{bmatrix}^\top$. In the following, we apply Theorem~\ref{thm:cl_stability_sos} to design a stabilizing controller for~\eqref{eq:exmp_nonlinear_dynamics} based on noisy data as in~\eqref{eq:nonlinear_system_dynamics3}. To be precise, we seek for a controller stabilizing the unstable steady-state $\xi_s = (\pi,0)$. 

	For the data generation, we choose $N=2000$, $B_w=I$, and a noise bound $\wbar=\num{1e-4}$. As for the examples in Section~\ref{sec:dd_controller_design_examples}, we consider $Q_w=-I$, $S_w=0$, and $R_w=\wbar^2NI$.
	We use equidistant data points generated by the true system dynamics with sampling rate $T_s = 0.001$. Initial state, input and noise are sampled uniformly from $x(t_0)\in[-2,2]^2$, $u(t_k)\in[-10,10]$, and the ellipse $\|\what(t_k)\|_2\leq\wbar$ for $k=1,\dots,N$, respectively.	
	
	Solving the SOS problem in Theorem~\ref{thm:cl_stability_sos}, we obtain $u(x)$ as in \eqref{eq:stabilizing_input}. To stabilize $\xi_s$, we apply the control law $u(x-\Psi(\xi_s))$ to System \eqref{eq:exmp_nonlinear_dynamics}.
	\begin{figure}[t]
		\vspace{4pt}
		\centering
		\includegraphics[width=\linewidth]{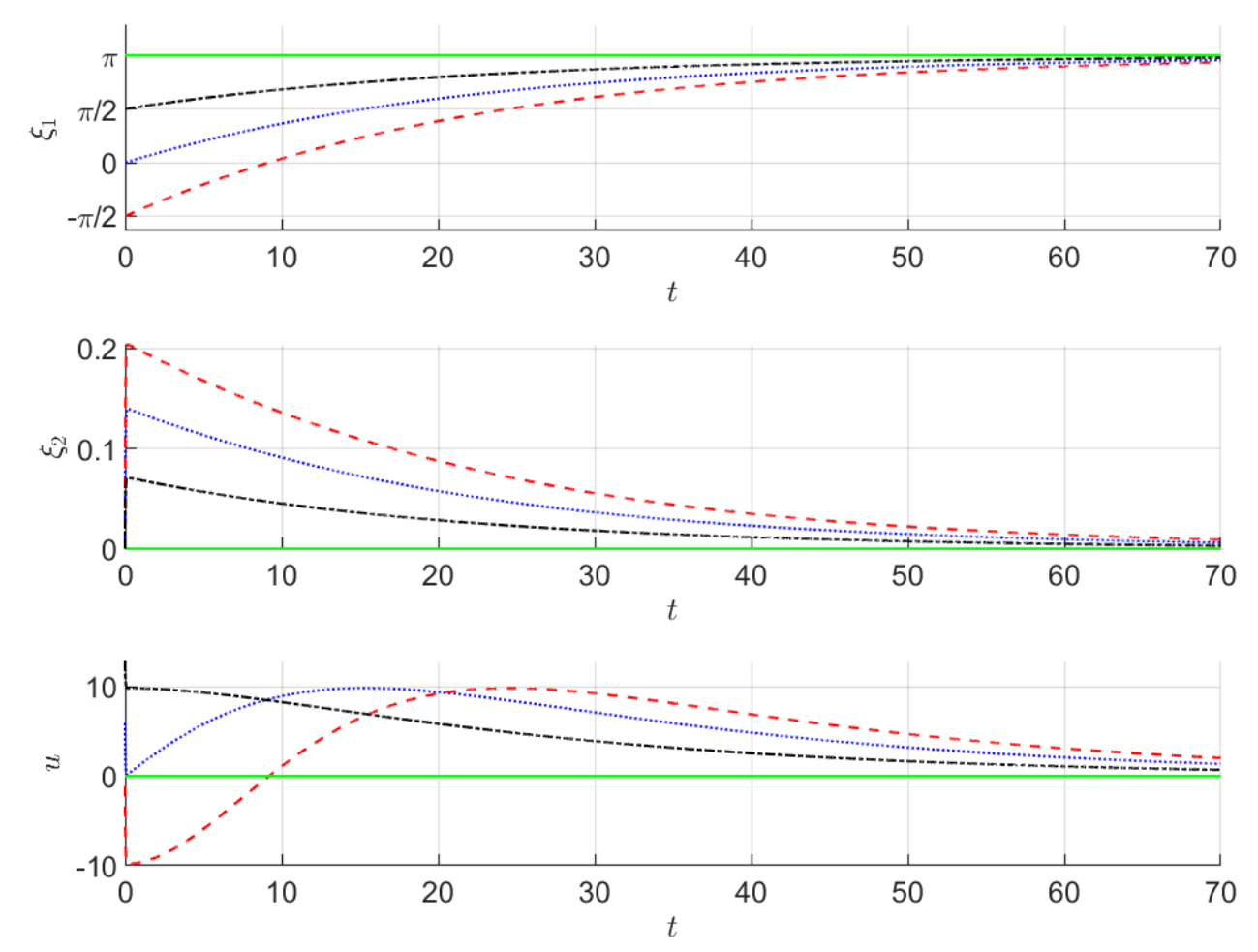}
		\vspace*{-22pt}
		\caption{Closed-loop simulations of System \eqref{eq:exmp_nonlinear_dynamics}.}
		\label{fig:exmp_pendulum}
		\vspace*{-16pt}
	\end{figure} 
	The resulting closed-loop behavior is depicted in Fig.~\ref{fig:exmp_pendulum} for different initial conditions. Note that the resulting controller successfully stabilizes the desired steady-state. For this example, \eqref{eq:cl_stability_condition_sos} is solved in less than one second.
\end{example}

\section{Conclusion}\label{sec:conclusion}
We proposed a data-driven control method with robust stability and performance guarantees for unknown nonlinear systems with possibly non-polynomial dynamics.
To this end, we exploited that we can rewrite the system dynamics in a structure linear in the unknown parameters and derived a stabilizing feedback controller using SOS methods. 
In our stability condition, we employed a parametrization of all systems consistent with the measured data and the introduced noise bound, and we extended the results to closed-loop performance. 
The number of decision variables of our approach is independent of the data length and thus, the procedure remains computationally tractable for medium-scale state-space dimensions. 
Since the framework also allows for non-polynomial basis functions, we were not only able to design controllers for rational systems but also nonlinear dynamics containing nonlinearities such as $\sin(x)$, $\sqrt{x}$, or $\exp(x)$. 
%

\bibliographystyle{IEEEtran}
\bibliography{local}

\end{document}